\documentclass[aps,amssymb,amsmath, twocolumn, pra, showpacs, superscriptaddress]{revtex4}
\usepackage{graphicx}
\usepackage{epsfig}
\usepackage{amsmath}
\usepackage{bm}
\usepackage{braket}
\usepackage{natbib}
\usepackage{enumitem}
\usepackage{longtable}

\usepackage[colorlinks, citecolor= magenta]{hyperref}
\usepackage[up]{subfigure}

\usepackage{epstopdf}

\usepackage{booktabs}

\usepackage{amsthm}
\newtheorem{thm}{Theorem}
\newtheorem{cond}{Condition}
\newtheorem{lem}{Lemma}
\newcounter{rtaskno}

\begin{document}

\title{ Generalized Formalism for Information Backflow in assessing Markovianity and its equivalence to Divisibility}
\author{Sagnik Chakraborty}
\email{csagnik@imsc.res.in}
\affiliation{Optics and Quantum Information Group, The Institute of Mathematical
Sciences, C. I. T. Campus, Taramani, Chennai 600113, India}
\affiliation{Homi Bhabha National Institute, Training School Complex, Anushakti Nagar, Mumbai 400094, India}

\begin{abstract}
We present a general framework for the information backflow (IB) approach of Markovianity that not only includes a large number, if not all, of IB prescriptions proposed so far but also is equivalent to CP-divisibility for invertible evolutions. Following the common approach of IB, where monotonic decay of some physical property or some information quantifier is seen as the definition of Markovianity, we propose in our framework a general description of what should be called a proper `physicality quantifier' to define Markovianity. We elucidate different properties of our framework, and use them to argue that an infinite family of non-Markovianity measures can be constructed, which would capture varied strengths of non-Markovianity in the dynamics. Moreover, we show that generalized trace-distance measure in $2$ dimension serve as a sufficient criteria for IB-Markovianity for a number of prescriptions suggested earlier in literature.
\end{abstract}


\maketitle

\section{Introduction}
Studying properties of open quantum system has recently received renewed focus from researchers working in various disciplines of physics. Interaction of a system with its environment gives rise to complex patterns of information flow between them and often this results in scenarios where the system evolution  retains memory of earlier times. This distinct property has been used to classify open quantum dynamics into two broad categories: {\it Markovian} or {\it memoryless} and {\it non-Markovian}. Although the classical analogue of this classification is well defined \cite{breuerbook,breuercolloquium,rhp_review},  the definition of Markovianity in quantum regime, is debated. There are numerous prescriptions which capture different aspects of this complex behavior, but a single unified description is yet to be found.

All the prescriptions suggested so far can be broadly classified into two main categories: {\it completely positive divisibility} (CPD) \cite{chru,rivas}, and {\it information backflow} (IB) \cite{breuer,wise1,wise,udevi,fisher1,temporalsteering,luo,interferometric,heLQU}. The CPD approach comes from a mathematical point of view, where a dynamical process is called Markovian if evolution up to a particular time $t$, can be broken down into two valid quantum evolutions: one up to an intermediate time $s$ (for any $s<t$), followed by another from $s$ to $t$. The IB approach, on the other hand describes a dynamical process to be Markovian if some {\it quantifier}, i.e. some physical property or some quantifier of information, decreases in a monotonic fashion under the action of the process. This approach identifies non-monotonic decrease of a quantifier, as backflow of information  from the environment to the system, which clearly is a signature of non-Markovianity. The IB category can be further divided into two classes based on the type of quantifier used. One, which uses quantifiers based on the system  only, which we call the {\it information backflow system only} (IBS) class, and the other that uses an ancilla to define the quantifier, which we call the {\it information backflow system-ancilla} (IBSA) class. The IBS class involves system based quantifiers like {\it distinguishability} \cite{breuer,wise1} and {\it generalized trace-distance} \cite{wise}. A number of measures to quantify the IBS class has also been suggested, namely, {\it fidelity} \cite{udevi}, {\it quantum fisher information} \cite{fisher1},  {\it temporal steering weight} \cite{temporalsteering}, etc. Any such measure can be used as a quantifier to define Markovianity. Similarly, the IBSA class involves joint system-ancilla quantifiers like {\it quantum mutual information} \cite{luo}, {\it interferometric power} \cite{interferometric}, {\it local quantum uncertainty} \cite{heLQU} etc.

With so many notions of Markovianity present, even within the IB category, one is compelled to look for inter-relations, hierarchies, or equivalence, that might be present within them. A number of studies \cite{comparison,haikkacomparing,Zengcomparison,chrucomparison} have already shown that these notions are not in general equivalent. To our knowledge, all the prescriptions under the IB category can be shown to be CPD but the converse is not true. Attempts have been made to find hierarchies \cite{chru_random,chenhierarchy} or equivalence \cite{bylickaconstructive,buscemiequivalence,chruscinski2017universal} within these approaches, either by concentrating on specific models or by using modified forms of some quantifier. But a general universal description that applies to any generic dynamics and any meaningful quantifier is not yet found.
 
In this paper, we attempt to tackle this problem by constructing a formalism that is independent of any particular form of quantifier. We give a generalized form of quantifier, called the {\it physicality quantifier} (PQ), for the whole of IB category. We set a minimal requirement criteria for any quantity to qualify as a PQ i.e. it should be non-increasing under any physical process. In doing so, we found that a large class, if not all, of quantifiers considered so far in the literature come as special cases of our generalized form. Basically we consider an ensemble  of quantum states, and define a  PQ as a real bounded function on the ensemble, that is non-increasing under any {\it completely positive} (CP) {\it trace preserving} (TP) map. We call a given dynamics, {\it  information backflow Markovian} or {\it IB-Markovian} if all possible PQ defined on any ensemble decreases in a monotonic fashion with time. IB-Markovianity is then sub-divided into {\it system only} (S) and {\it system-ancilla} (SA) class, where in the later case the ancilla considered, is of same dimension as the system. A dynamics is said to belong to S-Markovian class, if the choice of ensemble for PQ is restricted to include system states only, whereas it belongs to SA-Markovian class, if we allow arbitrary system-ancilla joint states in the ensemble. Later, we show that the S-Markovian class is a subset of the SA-Markovian class. Thus, SA-Markovianity comes out as an equivalent criteria to IB-Markovianity. Also, it can be easily inferred from the definition of PQ,  that our formalism is automatically CPD.

The motivation behind choosing this specific definition of PQ are two fold. On one hand, as we demonstrate here, it allows  for a lucid description of information flow between system and environment, which clearly shows backflow, whenever there is a departure from monotonicity in decay of a PQ. On the other hand, it results in filtering out of those quantities which are decreasing, even for isolated systems, under unitary evolution. Thus our definition serves as a minimal criteria, that captures only those quantities which reflect information exchange between system and environment. 

We then examine different properties of our formalism and prove that, for invertible dynamics, the following are equivalent: $(i)$ SA-Markovianity or IB-Markovianity, $(ii)$ Markovianity  with respect to generalized trace-distance measure (GTD) on an extended  system-ancilla space, and $(iii)$ CPD. We also show that our formalism can be used to construct an infinite family of non-Markovianity measures, which would capture varied strengths of memory effects present in the dynamics. Moreover, we prove that for qubit dynamics, GTD (defined only on the system) serves as a sufficient criteria for IB-Markovianity, for quantifiers defined on one or two system states; like quantum fisher information, fidelity, distinguishability, etc. Finally, we  present some applications of our formalism and discuss the context of our formalism with respect to incoherent, unital and non-invertible dynamics.

In Sec \ref{Preliminaries}, we present mathematical preliminaries, which describe notations and definitions used in the paper. In Sec \ref{Generalized formalism}, we propose our generalized formalism for information backflow, following which we discuss some general properties of it in Sec \ref{General properties}. In Sec \ref{qubit case} we prove certain results of our formalism for the qubit case. Finally, in Sec \ref{applications} we present some applications of our formalism, before discussing and concluding in Sec \ref{discussions} and section \ref{conclusion}.

\section{Preliminaries}
\label{Preliminaries}
If $\mathcal{H}$ is a Hilbert space, let $\mathcal{L}(\mathcal{H})$ be the space of all linear operators and   $\mathcal{P}_+(\mathcal{H})$ the set of all density matrices on $\mathcal{H}$. Let  $\mathcal{T}(\mathcal{H},\mathcal{H})$ denote the space of all linear maps from $\mathcal{L}(\mathcal{H})$ to $\mathcal{L}(\mathcal{H})$. Now, consider a $d-$dimensional system and a $d-$dimensional ancilla with Hilbert spaces $\mathcal{H_S}$ and $\mathcal{H_A}$, respectively.  A dynamical map $\Lambda_t\in\mathcal{T}(\mathcal{H_S},\mathcal{H_S})$ is a CPTP map describing evolution up to a  time  $t$. The full dynamics is described by a family of time-parametrized CPTP maps $\Lambda:=\{\Lambda_t\}_t$. By invertible dynamics, we mean a dynamical map $\Lambda_t$ which is invertible for all $t$. Note that, most of physical dynamical maps are invertible. Even the thermalization process, where any initial state evolves towards a fixed thermal state, is invertible for all finite times. 

{\bf Definition 1.} A dynamical map $\Lambda_t$ is said to be divisible if it can be expressed as,
\begin{equation}
\label{cpd}
\Lambda_t=V_{t,s}\Lambda_s,
\end{equation}
for any $ t>s$, where $V_{t,s}\in\mathcal{T}(\mathcal{H_S},\mathcal{H_S})$. If $V_{t,s}$ is (completely) positive for any $ t>s$, the dynamics is called (completely) positive divisible, and abbreviated as (C)PD.

Note that, $V_{t,s}$ represents intermediate evolution from $s$ to $t$, and it is uniquely defined only when  $\Lambda_t$ is invertible. Recently, Chruscinski et. al. \cite{chruscinski2017universal} showed that the necessary and sufficient condition for divisibility is, $Ker(\Lambda_s)\subseteq Ker(\Lambda_t)$ for any $t>s$, where $Ker(\Lambda)$ represents kernel of $\Lambda$. For a detailed mathematical characterization of divisibility and (C)PD, refer to \cite{chruscinski2017universal} . An ensemble on the system $\mathcal{E}_S:=\{p_i;\rho_i\}_{i=1}^n$ is defined as a finite collection of states $\rho_i\in\mathcal{P}_+(\mathcal{H_S})$ with {\it a priori} probabilities $p_i$. Similarly, we define  $\mathcal{E}_{SA}:=\{p_i;\xi_i\}_{i=1}^n$ on system-ancilla with $\xi_i\in\mathcal{P}_+(\mathcal{H_S}\otimes\mathcal{H_A})$. Let $\mathcal{F}_S^n:=\big\{\mathcal{E}_S~|~\mathcal{E}_S=\{p_i;\rho_i\}_{i=1}^n \big\}$ and $\mathcal{F}_{SA}^n:=\big\{\mathcal{E}_{SA}~|~\mathcal{E}_{SA}=\{p_i;\xi_i\}_{i=1}^n \big\}$ be the collection of all ensembles with $n$ elements. We define the set of all possible ensembles of any size by,
\begin{align}
\label{ensemble}
\mathcal{F}_S:=\bigcup\limits_{n=1}^{\infty}\mathcal{F}_S^n~~~;~~~
\mathcal{F}_{SA}:=\bigcup\limits_{n=1}^{\infty}\mathcal{F}_{SA}^n.
\end{align}
\begingroup
\begin{table}
\caption{Physicality Quantifiers and their Class and Type}
\begin{ruledtabular}
\begin{tabular}{c  c  c}
   {\bf ~Type} & {\bf Class~~} & {\bf Physicality quantifier} \\
  \hline
  $\mathcal{I}_S^1$ & $1$-S-Markovian & Quantum Fisher information \cite{fisher1} \\
  &&\\
  && Fidelity \cite{udevi} \\
  $\mathcal{I}_S^2$ & $2$-S-Markovian & State distinguishability \cite{breuer} \\
  && Generalized trace-distance \cite{wise} \\
  &&\\
  $\mathcal{I}_S^m$ & $m$-S-Markovian\footnote{If Alice makes measurement $M_{a|x}$, where $a=1,\dots,m_1$ and $x=1,\dots,m_2$, then $m=m_1m_2$ (see Appendix \ref{appendix physicality}).}  & Temporal steering weight \cite{temporalsteering}\\
  &&\\
  && Quantum mutual information \cite{luo}\\
  $\mathcal{I}_{SA}^1$ & $1$-SA-Markovian  &  Interferometric power \cite{interferometric} \\
  && Local quantum uncertainty \cite{heLQU}\\
  &&\\
  $\mathcal{I}_{SA}^2$ & $2$-SA-Markovian &  Generalized trace-distance extended  
 \end{tabular}
\end{ruledtabular}
\label{table1}

\end{table}
\endgroup

\section{Generalized formalism for information backflow}
\label{Generalized formalism}
We now present a generalized formalism for the IB category in such a way that a large class of IB prescriptions proposed so far \cite{breuer,wise1,wise,udevi,fisher1,temporalsteering,luo,interferometric,heLQU} falls into it. We identify an essential feature which is common to all quantifiers considered in the literature i.e. they are unitarily invariant and non-increasing under CPTP maps. This suggests a ready generalization of the definition of  quantifier. Therefore, we propose to define two types of PQ, $\mathcal{I}_S:\mathcal{F}_S\mapsto\mathbb{R}$ and $\mathcal{I}_{SA}:\mathcal{F}_{SA}\mapsto\mathbb{R}$, as real bounded functions on ensembles of quantum states, which follow  condition \ref{condition1}. For convenience,   any function $f_S(\mathcal{E}_S)$ or $f_{SA}(\mathcal{E}_{SA})$, is also represented by  symbols $f_S\big\{p_i;\rho_i\big\}$ or $f_{SA}\big\{p_i;\xi_i\big\}$, respectively.
\begin{cond}
\label{condition1}
 Let $T\in\mathcal{T}(\mathcal{H_S},\mathcal{H_S})$ be any CPTP map, acting on the system. For a given form of $\mathcal{I}_S$ or $\mathcal{I}_{SA}$, the following are true:
 \begin{align}
 \mathcal{I}_S\big\{p_i;T[\rho_i]\big\}&\leq\mathcal{I}_S\big\{p_i;\rho_i\big\},\nonumber\\
 \mathcal{I}_{SA}\big\{p_i;(T\otimes I)[\xi_i]\big\}&\leq\mathcal{I}_{SA}\big\{p_i;\xi_i\big\}\nonumber
 \end{align}
 where $I$ denotes the identity map in $\mathcal{T}(\mathcal{H_A},\mathcal{H_A})$.
\end{cond}
Note that, this readily implies invariance of PQ under any unitary evolution \cite{unitaryinv}, i.e.
\begin{align}
\mathcal{I}_S\big\{p_i;U\rho_iU^{\dagger}\big\}&=\mathcal{I}_S\big\{p_i;\rho_i\big\},\label{system unitary}\\ 
 \mathcal{I}_{SA}\big\{p_i;(U\otimes \mathbb{I})\xi_i(U^{\dagger}\otimes \mathbb{I})\big\}&=\mathcal{I}_{SA}\big\{p_i;\xi_i\big\},\label{system ancilla unitary}
 \end{align}
 where $U\in\mathcal{L}(\mathcal{H_S})$ and $\mathbb{I}\in\mathcal{L}(\mathcal{H_A})$ are unitary and identity operators, respectively.
As any PQ is real and bounded, for any given form we can always choose an equivalent form, by adding the lower bound, which would have the same monotonic or non monotonic nature, and also would be positive. We would therefore restrict ourselves to only positive PQ. Note that, distance measures like {\it p-norms}  \cite{pnorm} on qubit space, and general measures of information like  {\it order-$\alpha$} Renyi divergences  in any dimension for certain ranges of $\alpha$ \cite{renyi}, obey condition \ref{condition1}. We further sub-divide PQ according to $n$, i.e. number of elements present in the ensemble. We define special types of PQ, $\mathcal{I}_S^n$ and $\mathcal{I}_{SA}^n$, which are focused on ensembles of size $n$, in the following way: $\mathcal{I}_S^n(\mathcal{E}_S)=0$ and $\mathcal{I}_{SA}^n(\mathcal{E}_{SA})=0$, for any $\mathcal{E}_S\notin \mathcal{F}_S^n$ and $\mathcal{E}_{SA}\notin \mathcal{F}_{SA}^n$. Note that as $\mathcal{I}_S^n$ and $\mathcal{I}_{SA}^n$ are valid PQ, they obey condition \ref{condition1}. 

Observe that, different forms of quantifiers used in the literature, are defined on different  subsets of $\mathcal{F}_S^n$ or $\mathcal{F}_{SA}^n$, for various values of $n$. For example, GTD \cite{wise} is defined on all elements of $\mathcal{F}_S^2$, whereas distinguishability \cite{breuer} is defined only on those elements of $\mathcal{F}_S^2$ for which $p_1=p_2=1/2$. Likewise, quantum mutual information \cite{luo} is defined on all elements of $\mathcal{F}_{SA}^1$. To fit these  quantifiers in our formalism, we define compatible PQ  in each case, which are of the form $\mathcal{I}_S^n$ or $\mathcal{I}_{SA}^n$. For example for GTD, we define a PQ, $\mathcal{I}_S^{GTD}$, which takes the same value as GTD for elements in $\mathcal{F}_S^2$, and zero otherwise. Similarly for distinguishability, we define $\mathcal{I}_S^{BLP}$  such that, $\mathcal{I}_S^{BLP}\{p_1=1/2, p_2=1/2,\rho_1,\rho_2\}=||\rho_1-\rho_2||_1$ and $\mathcal{I}_S^{BLP}(\mathcal{E}_S)=0$, for $\mathcal{E}_S\notin\mathcal{F}_S^2$ or $p_i\neq1/2$ . Note
that BLP stands for Breuer, Laine, and Piilo, who were the first
to use distinguishability as a measure of non-Markovianity \cite{breuer}. In a similar way, we find that a large class of system-ancilla quantifiers considered so far \cite{luo,interferometric,heLQU} also correspond to PQ of the form $\mathcal{I}_{SA}^n$, for different values of $n$. In particular, $\mathcal{I}_S^2$ and $\mathcal{I}_{SA}^1$ corresponds to a large number of cases in the literature (see Table \ref{table1}). Refer to  Appendix \ref{appendix physicality}, for a detailed disposition of how each quantifier corresponds to PQ.

For a dynamical map $\Lambda_t$ we define {\it dynamic physicality quantifiers} $\Phi_t^{\mathcal{I}_S}$ and $\Phi_t^{\mathcal{I}_{SA}}$  based on $\mathcal{I}_S$ and $\mathcal{I}_{SA}$, in the following way,
\begin{align}
 \Phi_t^{\mathcal{I}_S}\big\{p_i;\rho_i\big\}&:=\mathcal{I}_S\big\{p_i;\Lambda_t[\rho_i]\big\},\label{system information quantifier}\\
 \Phi_t^{\mathcal{I}_{SA}}\big\{p_i;\xi_i\big\}&:=\mathcal{I}_{SA}\big\{p_i;(\Lambda_t\otimes I)[\xi_i]\big\}\label{system ancilla information quantifier}.
\end{align}
We now verify the perception, that non-monotonic decay of a PQ, rightly represents backflow of information from environment to the system. For any form of $\mathcal{I}_S$ or $\mathcal{I}_{SA}$, consider $\mathcal{I'}_S$  or $\mathcal{I'}_{SA}$ to be the same quantity defined on environment-system ensemble $\{p_i;\zeta^i_{ES}\}_i$ or environment-system-ancilla ensemble $\{p_i;\zeta^i_{ESA}\}_i$, respectively. As any open system dynamics is a result of unitary evolution $U_{ES}(t)$ of the system and environment considered jointly \cite{breuerbook}, we conclude from Eqs. (\ref{system unitary}) and (\ref{system ancilla unitary}), that $\Phi_t^{\mathcal{I'}_S}$ and $\Phi_t^{\mathcal{I'}_{SA}}$, defined in the sense of Eqs. (\ref{system information quantifier}) and (\ref{system ancilla information quantifier}), are constant in time. Now, consider $I^{\mathcal{I'}_S}_{env}(t)=\Phi_t^{\mathcal{I'}_S}-\Phi_t^{\mathcal{I}_S}$ or $I^{\mathcal{I'}_{SA}}_{env}(t)=\Phi_t^{\mathcal{I'}_{SA}}-\Phi_t^{\mathcal{I}_{SA}}$ to represent the information content of the environment and system (ancilla) combined, that cannot be obtained by knowing the system (ancilla) ensemble alone. Hence, we get $I^{\mathcal{I'}_S}_{env}(t)+\Phi_t^{\mathcal{I}_S}=constant$  and $I^{\mathcal{I'}_{SA}}_{env}(t)+\Phi_t^{\mathcal{I}_{SA}}=constant$, which means the net information content remains unchanged. Note, as partial tracing is always isomorphic to a CPTP map \cite{partialtrace}, from condition \ref{condition1}, we find $I^{\mathcal{I'}_S}_{env}(t)$ and $I^{\mathcal{I'}_{SA}}_{env}(t)$ are both positive quantities. Therefore, we conclude non-monotonic decay of $\Phi_t^{\mathcal{I}_S}$ and $\Phi_t^{\mathcal{I}_{SA}}$ rightly signifies
information backflow from environment to the system. Now we define Markovianity in terms of each valid form of PQ.

{\bf Definition 2.} A dynamical map $\Lambda_t$, is called {\it $\mathcal{I}_S$-Markovian} ({\it $\mathcal{I}_{SA}$-Markovian}) for some form of $\mathcal{I}_S$ ($\mathcal{I}_{SA}$), if  $\Phi_t^{\mathcal{I}_S}(\mathcal{E}_S)$ $\big(\Phi_t^{\mathcal{I}_{SA}}(\mathcal{E}_{SA})\big)$ decreases monotonically with time $t$, for any $\mathcal{E}_S\in\mathcal{F}_S$ ($\mathcal{E}_{SA}\in\mathcal{F}_{SA}$).

There are numerous examples in the literature for the above definition \cite{breuer,wise1,wise,udevi,fisher1,temporalsteering,luo,interferometric,heLQU}. We generalize the above notion in the following way,

{\bf Definition 3.} A dynamical map $\Lambda_t$ is called {\it $n$-S-Markovian} ({\it $n$-SA-Markovian}) if  $\Phi_t^{\mathcal{I}_S^n}(\mathcal{E}_S)$  $\big(\Phi_t^{\mathcal{I}_{SA}^n}(\mathcal{E}_{SA})\big)$ decreases in a monotonic fashion with time $t$, for any form of  $\mathcal{I}_S^n$ ($\mathcal{I}_{SA}^n$) and any choice of ensemble $\mathcal{E}_S\in\mathcal{F}_S^n$ ($\mathcal{E}_{SA}\in\mathcal{F}_{SA}^n$). 

We now give generalized definition of Markovianity for all PQ defined on system and system-ancilla. 

{\bf Definition 4.} A dynamical map $\Lambda_t$ is called {\it S-Markovian} ({\it SA-Markovian}) if  it is $n$-S-Markovian ($n$-SA-Markovian) for any value of $n$.

Finally, we give our generalized definition of Markovianity for backflow of information: any dynamics which is both S-Markovian and SA-Markovian, is called {\it IB-Markovian}.

\section{General properties of the formalism}
\label{General properties}
We first note a hierarchy within our Markovianity classes, which is apparent from their definition: any $n$-S-Markovian ($n$-SA-Markovian) class is a subset of $(n+1)$-S-Markovian ($(n+1)$-SA-Markovian) class. This observation provides a useful insight, that our formalism can be used to construct an infinite family of non-Markovianity measures, which would capture varied intensities of memory effects present in the dynamics. The higher the least value of $n$, for which a dynamics fails to be $n$-S-Markovian or $n$-SA-Markovian, the weaker is the effect of memory in the dynamics. Also note, as any PQ obeys condition \ref{condition1}, all IB-Markovian dynamics are automatically CPD. We now present a result, that makes SA-Markovianity an equivalent criteria to IB-Markovianity.
\begin{thm}
\label{thm0}
 If any dynamical maps $\Lambda_t$ is $n$-SA-Markovian, then it is  $n$-S-Markovian.
\end{thm}
This result is expected, as any PQ on the system can be seen as a PQ on system-ancilla by choosing the system ensemble states to be reduced density matrices of the system-ancilla ensemble states. See Appendix \ref{appendix theorem1} for detailed proof. Figure \ref{markov} gives a concise representation of all the hierarchies present in our formalism. We now show how IB-Markovianity is related to CPD and in what way GTD on extended space plays a vital role in relating these two quantities. GTD is a quantity that gives the best possible distinguishing probability of a pair of quantum states, occurring with different probabilities \cite{nielsen_chuang}. Suppose, Alice prepares one of two states $\rho_1$ and $\rho_2$ with probabilities $p_1$ and $p_2$ and sends the ensemble to Bob. The best possible probability for him to distinguish between these two states with a single-shot experiment is given by the GTD of the ensemble. The corresponding PQ,  $\mathcal{I}_S^{GTD}$ is given by, 
\begin{figure}
  \includegraphics[width= 8.6 cm,height= 5.5 cm]{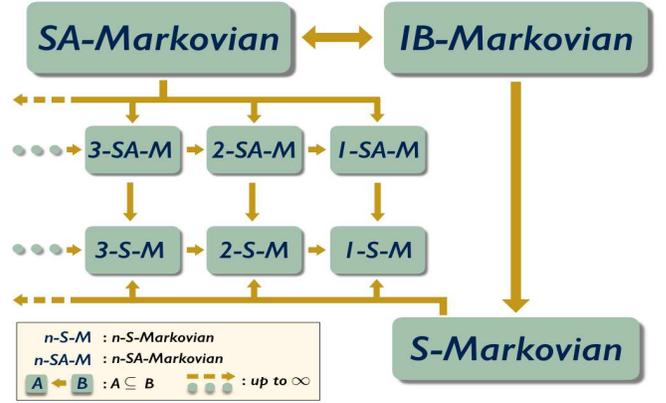}
  \caption{(colour online) Hierarchy of different classes of IB-Markovianity.}
  \label{markov}
 \end{figure}
\begin{equation}
\label{gtd quantifier}
 \mathcal{I}_S^{GTD}\{p_1,p_2,\rho_1,\rho_2\} := ||p_1\rho_1 - p_2\rho_2||_1,
\end{equation}
 where $||A||_1 = Tr\sqrt{A^{\dagger} A}$. It was first proposed in \cite{chrucomparison} and later in \cite{wise} that GTD can be used as a  quantifier to define Markovianity. The definition of GTD can also be easily extended to system-ancilla space. We call it {\it generalized trace-distance measure extended} (GTDE) and define it in the following way,
\begin{equation}
 \label{gtde quantifier}
 \mathcal{I}_{SA}^{GTDE}\{p_1,p_2,\xi_1,\xi_2\}:=||p_1\xi_1-p_2\xi_2||_1,
\end{equation}
where $\xi_i\in\mathcal{P}_+(\mathcal{H_S}\otimes\mathcal{H_A})$. Note that $\mathcal{I}_S^{GTD}$ and $\mathcal{I}_{SA}^{GTDE}$ are special forms of $\mathcal{I}_S^2$ and $\mathcal{I}_{SA}^2$, respectively.

{\bf Definition 5.} A dynamical map $\Lambda_t$ is called {\it GTD-Markovian} or {\it GTDE-Markovian} in the sense of definition $2$, if the respective PQ are $\mathcal{I}_S^{GTD}$ and $\mathcal{I}_{SA}^{GTDE}$.

Now we present one of the main theorems of this paper.
\begin{thm}
\label{thm1}
 For an invertible dynamical map $\Lambda_t$, the following are equivalent: $(i)$ $\Lambda_t$ is GTDE-Markovian, $(ii)$ $\Lambda_t$ is CPD, $(iii)$ $\Lambda_t$ is SA-Markovian or IB-Markovian, $(iv)$ $\Lambda_t$ is $2$-SA-Markovian.
\end{thm}
The non-intuitive part of the theorem is $(i)\implies(ii)$, which can be easily deduced by using a result by Kossakowski \cite{kossakowski1972quantum}, where TP contraction of trace-norm was shown to be the necessary and sufficient criteria for  positivity of maps (see Appendix \ref{appendix theorem2}). Thus we see the GTDE criteria is not only necessary but also sufficient for the whole of IB-Markovianity class, for invertible dynamics. 

\section{For the qubit case}
\label{qubit case}
We now show that for  qubit dynamics the GTD criteria is an equivalent criteria for $2$-S-Markovianity. From theorem \ref{thm0}, this implies GTD is also a sufficient criteria for $1$-S-Markovianity. Therefore, in qubit dynamics GTD serves as a sufficient criteria of Markovianity for quantifiers defined on one or two states of the system, namely, distinguishability \cite{breuer}, fidelity \cite{udevi}, quantum fisher information \cite{fisher1}, etc. The existence of this result is due to Alberti and Uhlmann \cite{alberti} and later also by Chefles et. al \cite{chefles} and Huang et. al. \cite{huang}, who showed that TP contractivity of trace-norm is necessary and sufficient condition for existence of physical transformations between two pairs of qubit states. The following lemma, and consequently the theorem can be easily deduced by applying this result (see Appendix \ref{appendix lemma1}). For this section of the paper, we assume $\mathcal{H_S}$ and $\mathcal{H_A}$ to be a qubit system and qubit ancilla. 
\begin{lem}
\label{lem1}
 If a qubit dynamical map $\Lambda_t$ is GTD-Markovian, then for any $t>s$ and any collection of states $\rho_1,\rho_2,\sigma_1,\sigma_2\in\mathcal{P}_+(\mathcal{H_S})$ such that $\rho_i = \Lambda_s[\sigma_i]~;~i=1,2$, there exists a CPTP map $T_{12}\in\mathcal{T}(\mathcal{H_S},\mathcal{H_S})$ such that,
 \begin{equation}
 \label{alberti condition}
  V_{t,s}[\rho_i] = T_{12}[\rho_i]~;~i=1,2.
 \end{equation}
\end{lem}
\begin{thm}
\label{thm2}
 A qubit dynamical map $\Lambda_t$ is GTD-Markovian if and only if it is $2$-S-Markovian.
\end{thm}
\begin{proof}
 As $\Lambda_t$ is GTD-Markovian, it can be expressed as Eq. (\ref{cpd}) (see Proposition 2 of \cite{chruscinski2017universal}). For any form of physicality quantifier $\mathcal{I}_S^2$, let us choose any $\rho_1,\rho_2,\sigma_1,\sigma_2\in\mathcal{P}_+(\mathcal{H_S})$ and $t>s$, such that $\rho_i = \Lambda_s[\sigma_i]$ for $i=1,2$. As $\Lambda_t$ is GTD-Markovian, using lemma \ref{lem1} we get $\mathcal{I}_S^2\{p_1,p_2,V_{t,s}[\rho_1],V_{t,s}[\rho_2]\}=\mathcal{I}_S^2\{p_1,p_2,T_{12}[\rho_1],T_{12}[\rho_2]\}$. Note, as $T_{12}$ is CPTP and $\mathcal{I}_S^2$ obeys Condition \ref{condition1}, this implies $\Phi_t^{\mathcal{I}_S^2}\leq\Phi_s^{\mathcal{I}_S^2}$ for any initial ensemble $\{p_1,p_2,\sigma_1,\sigma_2\}$ and $t>s$. Hence, we conclude if $\Lambda_t$ is GTD-Markovian then it is $2$-S-Markovian. The converse statement is easy to prove, as $\mathcal{I}_S^{GTD}$ is a physicality quantifier of the form $\mathcal{I}_S^2$. 
\end{proof}
It was shown in \cite{wise}, that GTD is an equivalent criteria to P-divisibility for invertible dynamics. Therefore, the above theorem shows that P-divisibility is equivalent to $2$-S-Markovianity for invertible qubit dynamics. As a detailed algorithm for constructing $T_{12}$ of lemma \ref{lem1} is already given in theorem 2.1 and 2.2 of \cite{huang}, we do not give any example, explaining the construction.

\section{Applications of the formalism}
\label{applications}
 \subsection{Minimum strength of Non-Markovianity required, to be used as a resource: Case studies}
 A number of protocols have been suggested, where backflow of information in a non-Markovian process has been used as a resource to enhance the efficiency of the protocol. All these protocols require different minimum strengths of non-Markovianity to enable the enhancement of efficiency. We now present two such scenarios and in each case we identify the minimum strength of non-Markovianity required.
 
 \subsubsection{Preserving channel capacities}
Non-Markovianity was used by Bylicka et al \cite{channelcapacity,rhp_review} to preserve {\it channel capacity} over long channels. They used the fact that {\it classical} and {\it quantum channel capacities}, given by $C_c$ and $C_q$, show a non-monotonic decay over time, whenever the dynamical map is non-Markovian.
 \begin{align}
  C_c[\Lambda_t]&=\sup_{\rho} I(\rho,\Lambda_t),\\
  C_q[\Lambda_t]&=\sup_{\rho} I_c(\rho,\Lambda_t),
 \end{align}
 where $I(\rho,\Lambda_t)$ is the quantum mutual information between the initial and the time evolved state of the system, defined in the following way,
 \begin{equation}
  I(\rho,\Lambda_t):=S(\rho)+I_c(\rho,\Lambda_t).
 \end{equation}
Here $S(\rho)=-\rho\log\rho$ is the von Neumann entropy and $I_c$ is the quantum coherent information, given by,
\begin{equation}
 I_c(\rho,\Lambda_t)=S\big(\Lambda_t[\rho]\big)-S\big((\Lambda_t\otimes I)\big[\ket{\Psi}\bra{\Psi}\big]\big),
\end{equation}
where $\rho=Tr_A\big(\ket{\Psi}\bra{\Psi}\big)$ and $\ket{\Psi}$ is a purification of $\rho$ in a higher dimensional system-ancilla space with ancilla dimension same as the system. It can be easily seen that $I_c$ is a PQ of the form $\mathcal{I}_S^1$. {\it Therefore, in order to have a revival of channel capacities, the dynamics should be non-Markovian at least with respect to 1-S-Markovianity class.} Thus we see a very strong form of non-Markovianity is necessary to use it as a resource in preserving channel capacities.

 \subsubsection{As a thermodynamic resource}
 Recently, Bylicka et al \cite{bylicka2016thermodynamic} used a non-Markovian dynamics to obtain revival of extractable work from an $n-$qubit system. The main ingredient behind this result is non-increasing nature of quantum mutual information between system and ancilla, under action of arbitrary CPTP map on one side of system or ancilla. {\it Hence, we conclude that in order to obtain a revival of extractable work, the dynamical process must be non-Markovian at least with respect to 1-SA-Markovianity class.}  
 
\subsection{Relation to the problem of existence of physical transformations between states}
 The problem of whether there exists a physical transformation between two sets of quantum states, is a well researched topic with various partial and complete results available \cite{alberti,huang,heinosaari,chefles,gour2017}. We will follow the notation in \cite{chefles} and denote the existence of physical transformation between two sets of quantum states, each containing $n$ elements, by $\{\rho_1,\rho_2,\dots,\rho_n\}\implies\{\sigma_1,\sigma_2\dots,\sigma_n\}$. Formally speaking, this means there exists a CPTP map $T$ connecting them i.e. $T[\rho_i]=\sigma_i$ for all $i=1,\dots,n$. In \cite{heinosaari}, it was shown that this problem can be reformulated as a semidefinite programming problem and thus using convex optimization techniques \cite{boydconvex}, it  can be checked algorithmically. For any given dynamics $\Lambda_t$, it can be easily shown (in a similar way as in Theorem \ref{thm2}) that, if $\{\rho_1(s),\rho_2(s),\dots,\rho_n(s)\}\implies\{\rho_1(t),\rho_2(t),\dots,\rho_n(t)\}$ for any $t>s$ and $\{\rho_i\}_{i=1}^n$, then $\Lambda_t$ is n-S-Markovian. Note, here $\rho_i(t)=\Lambda_t[\rho_i]$ represents the time evolved states.
 
 \subsection{A family of new non-Markovianity measures}
 As this formalism provides a general structure for constructing PQ, it is expected that a number of new PQ will emerge, that was previously unknown to the literature. Any such PQ can be used to device a non-Markovianity measure in the following way,
 \begin{equation}
  \mathcal{N}_{\mathcal{I}_S}(\Lambda_t):=\int_{\frac{d\big(\Phi_t^{\mathcal{I}_S}\big)}{dt}>0}\frac{d\big(\Phi_t^{\mathcal{I}_S}\big)}{dt}dt.
 \end{equation}
 Here $\Phi_t^{\mathcal{I}_S}$ is, as defined in Eq. (\ref{system information quantifier}). Note, positive time derivative of $\Phi_t^{\mathcal{I}_S}$ implies, departure from monotonic decay of the quantity over time. Similarly for PQ's of the form $\mathcal{I}_{SA}$, we can define non-Markovianity measures in the same way. Thus our formalism provides a platform for an infinite family of non-Markovianity measures. For example, in \cite{petz} a family of new metrics $g_D(A,B)$ on the space of linear operators of finite dimension were suggested, which are monotonic (decreasing) under stochastic (CPTP) maps i.e. $g_D(T[A],T[A])\leq g_D(A,A)$, for any CPTP map $T$, any operator $A$ and positive operator $D$. See \cite{petz} for more details about $g_D(A,B)$. Any such metric can used to device a new PQ of the form,
 \begin{equation}
 \mathcal{I}^{g_D}_S\{\rho\}=g_D(\rho,\rho). 
 \end{equation}
 Note this is a PQ of the form $\mathcal{I}^1_S$. Also, \cite{pnorm} presents a collection of norms, the $p-$norms, that can used as PQ on qubit space in the following way,
 \begin{align}
  \mathcal{I}_S^{p-norm}\{q_1,q_2,\rho_1,\rho_2\} &= ||\rho_1 - \rho_2||_p;~q_1=q_2=1/2\nonumber\\
  &=0~~~~~~~~~~~~~;~ q_i\neq 1/2,
 \end{align}
 where $||A||_p=[Tr(A^{\dagger}A)^{p/2}]^{1/p}$, $p\geq1$, $\rho_1,\rho_2\in\mathcal{P}_+(\mathcal{H}_S)$ with $\mathcal{H}_S=\mathbb{C}^2$ and $q_1,q_2$ are probabilities. Note, $p=1$ gives our usual trace-norm and we have, $\mathcal{I}_S^{1-norm}=\mathcal{I}_S^{BLP}$. We now combine two forms of $p-$norms to define a new PQ,
 \begin{equation}
 \label{p1-p2}
  \mathcal{I}^{p_1-p_2}_S(\mathcal{E}_S):=q_1\mathcal{I}_S^{p_1-norm}(\mathcal{E}_S')+ q_2\mathcal{I}_{S}^{p_2-norm}(\mathcal{E}_S')
  \end{equation}
 where $\mathcal{E}_S=\{q_1,q_2,\rho_1,\rho_2\}$ and $\mathcal{E}_S'=\{1/2,1/2,\rho_1,\rho_2\}$.  Note that $p-$norms defined on qubit space obey condition \ref{condition1} and hence $\mathcal{I}_{S}^{p-norm}$ and $\mathcal{I}_{S}^{p_1-p_2}$ qualify as PQ of the form $\mathcal{I}^2_S$.
 
To present an example, we consider the  the random unitary dynamics $\Lambda_t$ on a qubit system, which has been extensively studied in the literature. We follow the notations and results presented in \cite{chru_random} to test our new PQ. The random unitary dynamics is given by,
 \begin{equation}
 \label{random_unitary}
  \Lambda_t[\rho]=\sum_{\alpha=0}^3r_{\alpha}(t)\sigma_{\alpha}\rho\sigma_{\alpha},
 \end{equation}
 where $r_{\alpha}(t)$ are time dependent probabilities with $r_0(0)=1$, $\sigma_0=\mathbb{I}$ and $\sigma_1,\sigma_2$ and $\sigma_3$ are the Pauli spin matrices. Also note, the Pauli matrices happen to be the eigenvectors of $\Lambda_t$ with time-dependent eigenvalues $\lambda_i(t)$, i.e.  $ \Lambda_t[\sigma_i]=\lambda_i(t)\sigma_i$ for $i=0,\dots,3$. We find, $\lambda_i=\sum_{j=1}^4H_{ij}r_j$, where $H$ is the Hadamard matrix given by, 
 \begin{displaymath}
  H=\left(\begin{array}{cccc}
         1 & 1 & 1 & 1\\
         1 & 1 & -1 & -1\\
         1 & -1 & 1 & -1\\
         1 & -1 & -1 & 1\\
        \end{array}\right).
 \end{displaymath}
 Note $\lambda_0(t)=1$ for all $t$. The master equation of this dynamics is given by,
 \begin{equation}
  \frac{d}{dt}\rho_t=\sum^3_{i=0}\gamma_i(t)\sigma_i\rho_t\sigma_i,
 \end{equation}
 where $\gamma_i=\frac{1}{4}\sum_{j=0}^3H_{ij}\frac{\dot{\lambda}_j(t)}{\lambda_j(t)}$ for all $i=0,\dots,3$. This readily implies,
 \begin{equation}
    \sum_{i=0}^3\gamma_i(t)=0.
 \end{equation} 
 Note that, this identity simplifies the above master equation to the following form,
\begin{equation}
  \frac{d}{dt}\rho_t=\sum^3_{i=1}\gamma_i(t)(\sigma_i\rho_t\sigma_i-\rho_t),
 \end{equation} 
 Also, we find the following,
 \begin{align}
  \lambda_1(t)&=e^{-2[\Gamma_2(t)+\Gamma_3(t)]},\\
  \lambda_2(t)&=e^{-2[\Gamma_1(t)+\Gamma_3(t)]},\\
  \lambda_3(t)&=e^{-2[\Gamma_1(t)+\Gamma_2(t)]},
 \end{align}
 where $\Gamma_k(t)=\int_0^t\gamma_k(\tau)d\tau$, for $k=1,2,3$. Now, consider the ensembles $\mathcal{E}_S$ and $\mathcal{E}'_S$, used in Eq. (\ref{p1-p2}).  As $\rho_1 - \rho_2$ is a traceless hermitian operator, we get $\rho_1 - \rho_2=\sum_{k=1}^3x_k\sigma_k$, where $x_1,x_2$ and $x_3$ are real numbers. Thus we have,
 \begin{equation}
  ||\Lambda_t[\rho_1 - \rho_2]||_p=2^{1/p} \eta(t),
 \end{equation}
 where $\eta(t)=\sqrt{\sum^3_{k=1}\lambda_k(t)^2x_k^2}$. This implies,
 \begin{equation}
  \frac{d}{dt}||\Lambda_t[\rho_1 - \rho_2]||_p=\frac{2^{1/p-1}}{\eta(t)}\sum^3_{k=1}x_k^2\frac{d}{dt}|\lambda_k(t)|^2.
 \end{equation}
 Thus we have,
 \begin{align}
  \frac{d\Big(\Phi_t^{\mathcal{I}^{p_1-p_2}_S}(\mathcal{E}_S)\Big)}{dt}&=\Big[q_12^{1/p_1-1}+q_22^{1/p_2-1}\Big]\nonumber\\
  &\times\frac{d\Big(\Phi_t^{\mathcal{I}^{BLP}_S}(\mathcal{E}_S')\Big)}{dt},
 \end{align}
 where,
 \begin{equation}
   \frac{d\Big(\Phi_t^{\mathcal{I}^{BLP}_S}(\mathcal{E}_S')\Big)}{dt}=\frac{1}{\eta(t)}\sum^3_{k=1}x_k^2\frac{d}{dt}|\lambda_k(t)|^2. 
 \end{equation}
 So we see $\mathcal{I}^{p_1-p_2}_S$ witnesses non-Markovianity whenever $\mathcal{I}^{BLP}_S$ witnesses the same, and vice-versa. \\ 
 
 {\it Example 1.} On choosing  $\gamma_1(t)=\gamma_2(t)=1$ and $\gamma_3(t)=\sin t$, we get $\Gamma_1(t)=\Gamma_2(t)=t$ and $\Gamma_3(t)=1-\cos t$. So we have,
 \begin{align}
  \lambda_1(t)&=e^{-2(1+t-\cos t)},\\
  \lambda_2(t)&=e^{-2(1+t-\cos t)},\\
  \lambda_3(t)&=e^{-4t}.
 \end{align}
 Thus we get,
 \begin{align}
  \Phi_t^{\mathcal{I}^{p_1-p_2}_S}&(\mathcal{E}_S)=\big(q_12^{1/p_1-1}+q_22^{1/p_2-1}\big)\nonumber\\
  &\times\big[(x_1^2+x_2^2)~e^{-4(1+t-\cos t)}+x_3^2~e^{-8t}\big]^{\frac{1}{2}}
 \end{align}
 Since $q_1,q_2,p_1$ and $p_2$ are all positive and $x_1,x_2$ and $x_3$ are real, it can be easily seen that the above function is monotonically decreasing with $t$. Hence, we conclude the above dynamics is $\mathcal{I}^{p_1-p_2}_S-$Markovian for any $p_1,p_2\geq1$. Also, note this dynamics is not CPD in general, as $\gamma_3(t)$ can take negative values \cite{chru_random}. \\

 {\it Example 2.} Choose $r_1(t)=r_2(t)=\frac{1-r_0(t)}{4}$ and $r_3(t)=\frac{1-r_0(t)}{2}$. Therefore $\lambda_1(t)=\lambda_2(t)=\frac{3r_0(t)-1}{2}$ and $\lambda_3(t)=r_0(t)$. Also, $\gamma_1(t)=\gamma_2(t)=-\frac{\dot{r}_0(t)}{4r_0(t)}$ and $\gamma_3(t)=-\frac{(3r_0(t)+1)}{4r_0(t)}\frac{\dot{r}_0(t)}{(3r_0(t)-1)}$. So, we have,
 
 \begin{figure}
  \includegraphics[width= 8.6 cm,height= 5.5 cm]{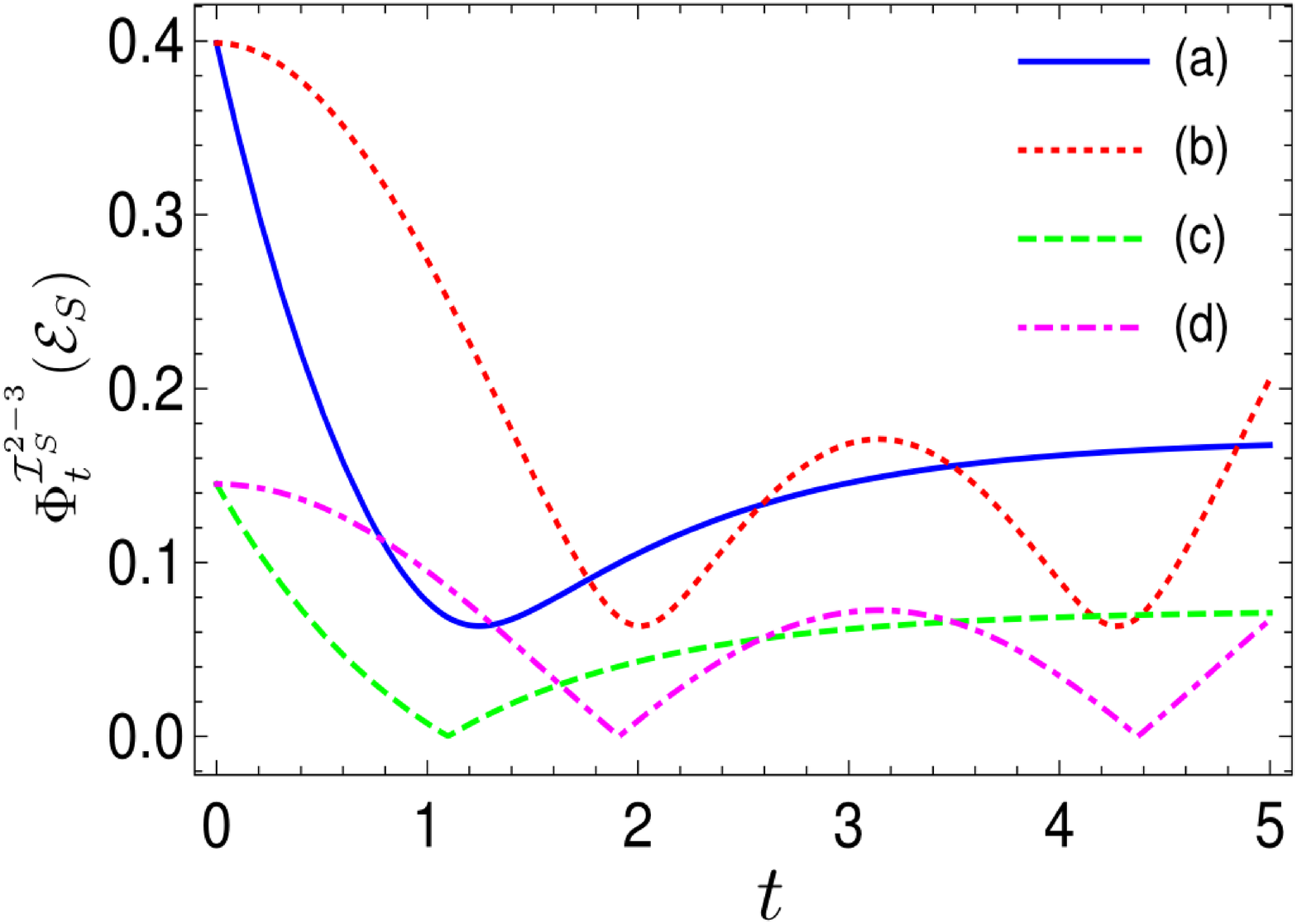}
  \caption{(colour online) Plot of dynamic PQ vs time for random unitary dynamics with $r_1(t)=r_2(t)=\frac{1-r_0(t)}{4}$ and $r_3(t)=\frac{1-r_0(t)}{2}$, for different initial ensembles $\mathcal{E}_S$ and different functional forms of $r_0(t)$. We consider $\mathcal{E}_S=\{q_1,q_2,\rho_1,\rho_2\}$, where $\rho_i=\frac{1}{2}(\sigma_0+\sum_{k=1}^3n^i_k\sigma_k)$ and $n^i=(s_i\sin\theta_i\cos \phi_i~,~s_i\sin\theta_i\sin \phi_i~,~s_i\cos\theta_i)$. The different cases considered here are: (a) $r_0(t)=e^{-t}$, $q_1=0.7,q_2=0.3$, $s_1=1,\theta_1=\phi_1=\pi/2$, $s_2=0.6,\theta_1=\pi,\phi_1=0$; (b) $r_0(t)=\frac{1+\cos t}{2}$, $q_1=0.7,q_2=0.3$, $s_1=1,\theta_1=\phi_1=\pi/2$, $s_2=0.6,\theta_1=\pi,\phi_1=0$; (c) $r_0(t)=e^{-t}$, $q_1=0.3,q_2=0.7$, $s_1=0.7,\theta_1=2\pi/3,\phi_1=\pi/6$, $s_2=0.4,\theta_1=5\pi/6,\phi_1=\pi/3$; (d) $r_0(t)=\frac{1+\cos t}{2}$, $q_1=0.3,q_2=0.7$, $s_1=0.7,\theta_1=2\pi/3,\phi_1=\pi/6$, $s_2=0.4,\theta_1=5\pi/6,\phi_1=\pi/3$.}
  \label{random}
 \end{figure}
 
 \begin{align}
  \Phi_t^{\mathcal{I}^{p_1-p_2}_S}&(\mathcal{E}_S)=\big(q_12^{1/p_1-1}+q_22^{1/p_2-1}\big)\nonumber\\
  &\times\Big[\frac{1}{4}\big(x_1^2+x_2^2\big)\big(3r_0(t)-1\big)^2+x_3^2~r_0(t)^2\Big].
 \end{align}
 Figure \ref{random} shows time evolution of $\Phi_t^{\mathcal{I}^{2-3}_S}$ for different initial ensembles and two different forms of $r_0(t)$: $e^{-t}$ and $\frac{1+\cos t}{2}$. We find in both cases the dynamics is $\mathcal{I}^{2-3}_S-$ non-Markovian.
 
\section{discussions}
\label{discussions}

{\it Incoherent and Unital dynamics.}$-$
A number of measures of non-Markovianity were suggested for incoherent and unital dynamics \cite{l1norm,he2017non,rentropa}, which are not non-increasing under arbitrary CPTP maps. Therefore, they are not guaranteed to be monotonically decreasing under arbitrary CPD dynamics, which are well accepted Markovian dynamics. Hence, in spite of them being useful non-Markovianity measures for certain types of dynamics, we do not consider them as appropriate quantifiers for describing Markovianity, in general. 

{\it Non-invertible dynamics.}$-$
As mentioned earlier, the form of $V_{t,s}$ is not unique when $\Lambda_t$ is non-invertible. Naturally, for any given non-invertible and divisible dynamics, different forms of  $V_{t,s}$ would correspond to different physical realizations of the dynamics. Recently, it was shown in \cite{chruscinski2017universal} that if a (also non-invertible) dynamics is GTDE-Markovian, there exists at least one form of $V_{t,s}$, and consequently a physical realization of the dynamics, that would be CPD. Therefore, it can be inferred that at least one realization of every non-invertible dynamics, obeys theorem \ref{thm1}. 

Recently, it was brought to our notice, that there has been attempts to characterize information flux for defining Markovianity, from a thermodynamic perspective \citep{chen2017}. The motivation behind this approach is similar to ours, i.e. to reconcile the conceptual difficulty arising from numerous non-equivalent definitions of IB.

\section{Conclusion}
\label{conclusion}
In this paper, we have provided a generalized formalism for describing the IB approach of Markovianity. We provided a general form of a quantifier, called the physicality quantifier, whose monotonic decay with time was seen as the defining criteria for Markovianity. We defined the physicality quantifier, to be any real bounded function on the ensemble space, that is non-increasing under CPTP maps. In doing so, we found that a large number of prescriptions for IB-Markovianity in the literature, come as special cases of our formalism. Also, by using our formalism we showed that for invertible dynamics, IB-Markovianity is equivalent to CP-divisibility, as well as to Markovianity with respect to generalized trace-distance measure in extended system-ancilla space.

We showed hierarchies of different subclasses of our formalism and argued that it can be used to construct an infinite family of non-Markovianity measures, which would capture varied strengths of memory effects present in the dynamics. We also used the formalism to show that generalized trace-distance measure for qubit dynamics, serve as sufficient criteria of IB-Markovianity for a number of prescriptions suggested earlier. Finally, we presented certain applications of our formalism.
We expect our formalism will shed light into further understanding of physical and mathematical structure of quantum Markovianity and enhance its applicability to more varied scenarios. 

\begin{acknowledgments}
The author would like to thank Francesco Buscemi for useful suggestions on the paper.  SC is also thankful to Sibasish Ghosh and Bassano Vacchini for useful discussions on non-Markovianity and Suchetana Goswami, Arindam Mallick and Arpan Das for reading the manuscript and suggesting appropriate alterations.
\end{acknowledgments}

\appendix
\section{Physicality quantifiers considered so far in Literature}
\label{appendix physicality}
Most of the quantifiers, suggested in the literature, are defined on ensembles having fixed number of elements. To fit them as valid physicality quantifiers, which are defined on ensembles of any size, we define special forms of physicality quantifiers $\mathcal{I}_S^n$ or $\mathcal{I}_{SA}^n$, which are focused on  ensembles of size $n$. We define,  $\mathcal{I}_S^n(\mathcal{E}_S)=0$ and $\mathcal{I}_{SA}^n(\mathcal{E}_{SA})=0$, for any $\mathcal{E}_S\notin \mathcal{F}_S^n$ and $\mathcal{E}_{SA}\notin \mathcal{F}_{SA}^n$. We now show that  a large number of quantifiers considered so far, correspond to physicality quantifiers of the form $\mathcal{I}_S^n$ or $\mathcal{I}_{SA}^n$ (for various values of $n$) in such a way, that the physicality quantifier takes the same value as the quantifier, for ensembles of size $n$. We denote $\mathbb{Z}_k$ to be the set of positive integers from $1$ to $k$ and $[a,b]$ to be the closed interval of real numbers from $a$ to $b$.
\begin{enumerate}[label=(\roman*)]
 \item Breuer et. al. \cite{breuer} considered an equal mixture of states i.e. $p_1=p_2=1/2$ and defined 
 {\it distinguishability} of two states as their quantifier. We define the physicality quantifier as,
 \begin{align}
  \mathcal{I}_S^{BLP}\{p_1,p_2,\rho_1,\rho_2\} &= ||\rho_1 - \rho_2||_1;~p_1=p_2=1/2\nonumber\\
  &=0~~~~~~~~~~~~~;~p_i\neq 1/2,
 \end{align}
 where $||A||_1 = Tr\sqrt{A^{\dagger} A}$. The Markovianity criteria, corresponding to this quantifier, is popularly known as the {\it BLP-}criteria of Markovianity. It can be easily shown that the above quantity is bounded and non-increasing under CPTP maps \cite{nielsen_chuang}. Hence, we see $\mathcal{I}_S^{BLP}$ is a particular form of $\mathcal{I}_S^2$.
 
 \item Rajagopal et. al. \cite{udevi} also considered a form that corresponds to  $\mathcal{I}_S^2$. They took equal mixture of two states and used {\it fidelity} as their measure of non-Markovianity i.e $p_1=p_2=1/2$. We slightly modify their definition and define the physicality quantifier in the following form,
 \begin{align}
  \mathcal{I}_S^{Fid}\{p_1,p_2,\rho_1,\rho_2\} &= 1-||\sqrt{\rho_1}\sqrt{\rho_2}||_1;~p_1=p_2=1/2\nonumber\\
  &=0~~~~~~~~~~~~~~~~~~~;~ p_i\neq 1/2,
 \end{align}
 It is easy to show that $\mathcal{I}_S^{Fid}$ lies in the interval $[0,1]$ and non-increasing under CPTP maps \cite{udevi,nielsen_chuang}.

 \item Lu et. al \cite{fisher1} used {\it quantum fisher information} (QFI) as the quantifier.
 QFI is the coefficient of efficiency in estimation of some parameter, say $\theta$, encoded in quantum state $\rho_{\theta}$. It can be shown  that QFI is infinitesimal Bures distance between two states \cite{rhp_review}. Therefore, the corresponding physicality quantifier is defined as, 
 \begin{equation}
  \mathcal{I}_S^{QFI}(\rho_{\theta})=4\lim_{\delta \theta\rightarrow 0} \Big[\frac{D_B(\rho_{\theta+\delta\theta},\rho_{\theta})}{\delta \theta}\Big]^2,
 \end{equation}
 where $D_B(\rho_1,\rho_2)=\sqrt{2\big[1-||\sqrt{\rho_1}\sqrt{\rho_2}||_1\big]}$.  We expect $\rho_{\theta}$ to be well behaved function of $\theta$, in the sense that it is differentiable. Therefore, we expect QFI to be a bounded function. Also, it is easy to see from the definition, that QFI is  non-increasing under CPTP maps \cite{fisher1,rhp_review,fujiwara2001}. Hence, 
 $\mathcal{I}_S^{QFI}$ is a physicality quantifierof the form $\mathcal{I}_S^1$.

 \item Chen et. al. used {\it temporal steering weight} (TSW) \cite{temporalsteering} to quantify Markovianity. In this setting Alice performs measurement $M_{a|x}$ , ( $a\in\mathbb{Z} _{m_1}~,~x\in\mathbb{Z}_{m_2}$) on a system state , creating an ensemble $\{p(a|x),\rho_{a|x}\}_{a|x}$. The ensemble is then passed through a dynamical map $\Lambda_t$ and the TSW of the output ensemble is calculated at each instant $t$. It was shown in \cite{temporalsteering} that TSW is non-increasing under CPTP maps. Also, TSW refers to maximum value of $\mu$ in Eq. (4) of \cite{temporalsteering}. Therefore,  it is evident from the construction that $0\leq\mu\leq 1$. Thus we see, TSW corresponds to a physicality quantifier of the form $\mathcal{I}_S^m$, where $m=m_1m_2$. 

 \item Dhar et. al. \cite{interferometric} used {\it interferometric power} for their criteria of Markovianity. It can be shown to be non-increasing under CPTP maps \cite{interferometric}. Also, for any system-ancilla state, interferometric power is calculated by optimizing over local unitary operations on the ancilla side \cite{interferometricprl}. Therefore, it can be inferred that interferometric power is bounded. Hence, interferometric power is a valid physicality quantifier and is of the form $\mathcal{I}_{SA}^1$ .

 \item He et. al. introduced a measure of non-Markovianity based on {\it local quantum uncertainty} (LQU) \cite{heLQU}. It can be shown that LQU is non-increasing under CPTP maps. Also note, as  LQU is determined though minimization of a bounded function over unitaries, it is evident that LQU is bounded \cite{lqucharacterize}. Hence we conclude, LQU corresponds to a physicality quantifier of the form $\mathcal{I}_{SA}^1$.
 
 \item Luo et. al. \cite{luo} used {\it quantum mutual information} (QMI) as their quantifier. The corresponding  physicality quantifier is
 \begin{equation}
  \mathcal{I}_{SA}^{QMI}(\xi)=S(\xi_S)+S(\xi_A)-S(\xi), 
 \end{equation}
 where $\xi\in\mathcal{P}_+(\mathcal{H_S}\otimes\mathcal{H_A})$, $\xi_{S/A}=Tr_{A/S}(\xi)$ are reduced density matrices and $S(\rho)=-\rho\log\rho$ is the usual von Neumann entropy. QMI is known to be bounded \cite{MMWilde} and non-increasing under CPTP maps \cite{dataprocessing}. Note that $\mathcal{I}_{SA}^{QMI}$ is a special form of $\mathcal{I}_{SA}^1$.
 
\end{enumerate}
\section{Detailed proof of Theorem 1}
\label{appendix theorem1}

\begin{proof}
 Assume $\Lambda_t$ is $n$-SA-Markovian. For any system based physicality quantifier $\mathcal{I}_S^n$, we define a real valued function $\mathcal{I}_{SA}^n$ on $\mathcal{F}_{SA}$, such that $\mathcal{I}_{SA}^n(\mathcal{E}_{SA})=\mathcal{I}_S^n(\mathcal{E}_S)$, where $\mathcal{E}_S=\{p_i;\rho_i\}$, $\mathcal{E}_{SA}=\{p_i;\xi_i\}$ and $Tr_A(\xi_i)=\rho_i$. This implies $\mathcal{I}_{SA}^n(\mathcal{E}_{SA})=0$ for $\mathcal{E}_{SA}\notin\mathcal{F}_{SA}^n$ and for any CPTP map $T\in\mathcal{T}(\mathcal{H}_S,\mathcal{H}_S)$, we get $\mathcal{I}_{SA}^n\{p_i;(T\otimes I)[\xi_i]\}=\mathcal{I}_S^n\{p_i;T[\rho_i]\}\leq\mathcal{I}_S^n\{p_i;\rho_i\}=\mathcal{I}_{SA}^n\{p_i;\xi_i\}$. Also note, as $\mathcal{I}_S^n$ is bounded, $\mathcal{I}_{SA}^n$ must also be bounded. Therefore, we see for any physicality quantifier $\mathcal{I}_S^n$ on the system, there exists a physicality quantifier  on system-ancilla, which is of the form $\mathcal{I}_{SA}^n$. These imply $\Phi_t^{\mathcal{I}_{SA}^n}\big\{p_i;\xi_i\big\}=\Phi_t^{\mathcal{I}_S^n}\big\{p_i;\rho_i\big\}$ (see Eqs. (\ref{system information quantifier}) and (\ref{system ancilla information quantifier})). Thus, monotonic decrease of $\Phi_t^{\mathcal{I}_{SA}^n}$ implies monotonic decrease of $\Phi_t^{\mathcal{I}_S^n}$. Thus, $\Lambda_t$ is $n$-S-Markovian.
\end{proof}

\section{Detailed proof of Theorem 2}
\label{appendix theorem2}
\begin{proof}
 Note as $\Lambda_t$ is invertible, it is also divisible i.e. it can be decomposed in the form of Eq. (\ref{cpd}).
 
 $(i)\implies(ii)$. Consider any hermitian operator  $H\in\mathcal{L}(\mathcal{H_S}\otimes\mathcal{H_A})$. As $\Lambda_s$ is an invertible and positive map for any $s>0$, there exists a hermitian operator $\tilde{H}$ such that $(\Lambda_s\otimes I)[\tilde{H}]=H$. Also from \cite{wise}, we know that any hermitian operator can be written as a positive number multiple of a Helstrom matrix i.e. $\tilde{H}=\lambda(p_1\xi_1-p_2\xi_2)$ for $\lambda>0$, $\xi_1,\xi_2\in\mathcal{P}_+(\mathcal{H_S}\otimes\mathcal{H_A})$ and probabilities $p_1,p_2$. As $\Lambda_t$ is GTDE-Markovianity, for any $t>s$, we get $||(V_{t,s}\otimes I)[H]||_1=||(\Lambda_t\otimes I)[\tilde{H}]||_1\leq||(\Lambda_s\otimes I)[\tilde{H}]||_1=||H||_1$ for any hermitian $H$. This implies $V_{t,s}\otimes I$ is a positive map for any $t>s$ \cite{kossakowski1972quantum}. Therefore, $V_{t,s}$ is CP for any $t>s$. Thus, $\Lambda_t$ is CPD.
 
 $(ii)\implies(iii)$. As $\Lambda_t$ is CPD, $V_{t,s}$ is CPTP. Hence from condition \ref{condition1} and Eq. (\ref{system ancilla information quantifier}), we get $\Phi_t^{\mathcal{I}_{SA}}=\mathcal{I}_{SA}\big\{p_i;(V_{t,s}\otimes I)(\Lambda_s\otimes I)[\xi_i]\big\}\leq\Phi_s^{\mathcal{I}_{SA}}$ for any form of $\mathcal{I}_{SA}$ and any ensemble $\mathcal{E}_{SA}=\{p_i;\xi_i\}$. Therefore, $\Lambda_t$ is SA-Markovian. Hence, from theorem \ref{thm0}, $\Lambda_t$ is IB-Markovian.
 
 $(iii)\implies(iv)$. This follows from the definition of SA-Markovianity.
 
 $(iv)\implies(i)$. This is trivial as $\mathcal{I}_{SA}^{GTDE}$ in Eq. (\ref{gtde quantifier}), given is of the form $\mathcal{I}_{SA}^2$.
\end{proof}

\section{Detailed proof of Lemma 1}
\label{appendix lemma1}

\begin{proof}
 Note, as $\Lambda_t$ is GTD-Markovian it is divisibile i.e. it can be expressed as Eq. (\ref{cpd}) (see Proposition 2 of \cite{chruscinski2017universal}). First Alberti et. al. \cite{alberti} and later Huang et. al. \cite{huang} showed that the necessary and sufficient condition for a collection of qubit states $\rho_1,\rho_2,\rho_1',\rho_2'\in\mathcal{P}_+(\mathcal{H_S})$ to have a CPTP map $T_{12}$ connecting them i.e. $T_{12}[\rho_i]=\rho_i'~;~i=1,2$, is $||\rho_1'-x\rho_2'||_1\leq||\rho_1-x\rho_2||_1$ for any $x\geq0$. Since, $p_1$ and $p_2$ in the GTD quantifier in Eq. (\ref{gtd quantifier}) are probabilities, without loss of generality we can choose $p_1>0$. For any $t>s$, let us now choose $\rho_i'=V_{t,s}[\rho_i]$ ; $i=1,2$, and $x=p_2/p_1$. Therefore,  if $\Lambda_t$ is GTD-Markovian, the necessary and sufficient condition for the existence of $T_{12}$ connecting $\rho_1,\rho_2,\rho_1'$ and $\rho_2'$  is satisfied for any $t>s$. Hence, we conclude there must exist a CPTP map $T_{12}$ satisfying Eq. (\ref{alberti condition}).
\end{proof}

\bibliography{gen_blp1.bib}

\begin{thebibliography}{47}
\bibitem[{par()}]{partialtrace}
\bibinfo{note}{Consider the map {$\Omega[X]=\sum_iK_iXK^{\dagger}_i$}, where
  {$K_i=\ket{\psi}\bra{i}_E\otimes\mathbb{I}_S$} and {$\{\ket{i}\}_i$} forms an
  orthonormal basis in environment space. This CPTP map is isomorphic to
  partial trace operation {$Tr_E(\bullet)$}}.
\end{thebibliography}

\end{document}